\documentclass[twoside,leqno,twocolumn]{article}  
\usepackage{ltexpprt} 
\usepackage{amsmath,amsfonts,amssymb}
\usepackage[usenames]{color}

\usepackage{bm}

\usepackage{xspace}

\usepackage{paralist}

\usepackage{graphicx}
\usepackage{epstopdf}

\usepackage{boxedminipage}

\usepackage{subfig}

\usepackage[colorlinks,urlcolor=blue,citecolor=blue,linkcolor=blue]{hyperref}

\usepackage{subfig}
\usepackage{tabularx}
\usepackage{array}

\usepackage{booktabs}
\usepackage{textcase}

\usepackage{multirow}

\usepackage{braket}

\usepackage{tikz}
\newcommand*\circled[1]{\tikz[baseline=(char.base)]{
            \node[shape=circle,draw,minimum size=3.5mm,inner sep=0pt] (char) {#1};}}

\usepackage{rotating}

\usepackage{algorithm}
\usepackage{algorithmic}
\usepackage{slashbox}

\def\submit{1}   
\ifnum\submit=1
	\newcommand{\full}[1]{}
	\newcommand{\confer}[1]{#1}
\else
	\newcommand{\full}[1]{#1}
	\newcommand{\confer}[1]{}
\fi

\newcommand{\cN}{{\cal N}}

\newcommand{\cV}{{\cal V}}

\newcommand{\qed}{\hfill $\Box$}

\newcommand{\cei}[1]{\lceil #1 \rceil}

\newcommand{\pr}{{\rm Pr}}
\newcommand{\Prob}[1]{\pr\left\{#1\right\}}

\newcommand{\eps}{\varepsilon}

\newcommand{\EX}{\hbox{\bf E}}

\newcommand{\Sec}[1]{\hyperref[sec:#1]{\S\ref*{sec:#1}}} %
\newcommand{\Eqn}[1]{\hyperref[eq:#1]{(\ref*{eq:#1})}} %
\newcommand{\Fig}[1]{\hyperref[fig:#1]{Fig.\,\ref*{fig:#1}}} %
\newcommand{\Tab}[1]{\hyperref[tab:#1]{Tab.\,\ref*{tab:#1}}} %
\newcommand{\Thm}[1]{\hyperref[thm:#1]{Thm.\,\ref*{thm:#1}}} %
\newcommand{\Lem}[1]{\hyperref[lem:#1]{Lem.\,\ref*{lem:#1}}} %
\newcommand{\Prop}[1]{\hyperref[prop:#1]{Prop.~\ref*{prop:#1}}} %
\newcommand{\Cor}[1]{\hyperref[cor:#1]{Cor.~\ref*{cor:#1}}} %
\newcommand{\Def}[1]{\hyperref[def:#1]{Defn.~\ref*{def:#1}}} %
\newcommand{\Alg}[1]{\hyperref[alg:#1]{Alg.\,\ref*{alg:#1}}} %
\newcommand{\Ex}[1]{\hyperref[ex:#1]{Ex.~\ref*{ex:#1}}} %
\newcommand{\Clm}[1]{\hyperref[clm:#1]{Claim~\ref*{clm:#1}}} %
\newcommand{\Step}[1]{\hyperref[step:#1]{Step~\ref*{step:#1}}} %

\newcommand{\gcc}{C}
\newcommand{\lcc}{\bar C}
\newcommand{\dcc}{{C}}

\newcommand{\algkappa}{$C$-{\tt wedge sampler}}
\newcommand{\alglcc}{$\lcc$-{\tt wedge sampler}}
\newcommand{\algdcc}{$\dcc_d$-{\tt wedge sampler}}
\newcommand{\algtrid}{$T_d$-{\tt wedge sampler}}

\newcommand{\Omit}[1]{}

\begin{document}

\title{Triadic Measures on Graphs: The Power of Wedge Sampling%
\thanks{This work was funded by the DARPA Graph-theoretic Research
    in Algorithms and the Phenomenology of Social Networks (GRAPHS)
    program and by the DOE ASCR Complex Interconnected Distributed Systems
    (CIDS) program, and Sandia's Laboratory Directed Research \& Development (LDRD) program. Sandia National Laboratories is a multi-program
    laboratory managed and operated by Sandia Corporation, a wholly
    owned subsidiary of Lockheed Martin Corporation, for the
    U.S. Department of Energy's National Nuclear Security
    Administration under contract DE-AC04-94AL85000.}}
\date{}

\author{C. Seshadhri\thanks{Sandia National Laboratories, CA, scomand@sandia.gov} 
\and
Ali Pinar\thanks{Sandia National Laboratories, CA, apinar@sandia.gov} 
\and 
Tamara G. Kolda\thanks{Sandia National Laboratories, CA, tgkolda@sandia.gov} }

\maketitle

\begin{abstract} 
Graphs are used to model interactions in a variety of contexts, and 
there is a growing need to quickly assess the structure  of a graph. 
Some of the most useful graph metrics, especially those measuring social cohesion,
are based on \emph{triangles}.
Despite the  importance of these triadic measures, associated algorithms  can be extremely expensive. 
We propose a new method based on \emph{wedge sampling}. This versatile technique allows for the fast and accurate approximation of all current variants of clustering coefficients and 
enables rapid uniform sampling of the triangles of a graph. 
Our methods come with  \emph{provable} and practical time-approximation tradeoffs for all computations. We provide extensive results that show our methods are orders of magnitude faster than the state-of-the-art, while providing nearly the accuracy of full enumeration. 
Our results will enable more wide-scale adoption of triadic measures for analysis of extremely large graphs, as demonstrated on several real-world examples.

\end{abstract}

\section{Introduction}

Graphs are used  to model  infrastructure networks,  the World Wide Web, computer traffic, molecular interactions, ecological systems, epidemics, citations, and social interactions, among others.  Despite the differences in the motivating applications, some topological structures have emerged to be important across all these domains. 
triangles, which can be explained by homophily (people become friends with those similar to themselves) and transitivity (friends of friends become friends).   This abundance of triangles,  along with  the information they reveal, motivates metrics such as the \emph{clustering coefficient} and the \emph{transitivity ratio}s~\cite{WaFa94, WaSt98}.
The triangle structure of a graph is commonly used in the social sciences for positing various theses
on behavior~\cite{Co88, Po98, Burt04, FoDeCo10}. 
Triangles have also been used  in graph mining applications such as spam detection and finding common topics on the WWW~\cite{EcMo02, BeBoCaGi08}.
The authors' earlier work used distribution of degree-wise clustering coefficients as the driving force for a new generative model, Blocked Two-Level Erd\"os-R\'enyi~\cite{SeKoPi11}.  
The authors' have also observed that  relationships among degrees of triangle vertices can be a descriptor of the underlying graph~\cite{DuPiKo12}.

\subsection{Clustering coefficients} \label{sec:info}
The information about triangles is usually summarized in terms of \emph{clustering coefficients}. 
Let $G$ be a simple undirected graph with $n$ vertices and $m$ edges.
Let $T$ denote the  number of triangles in the graph and $W$ be the number of \emph{wedges} (a path of length $2$).
The most common measure is the \emph{global clustering coefficient} $\gcc = 3T/W$, which 
measures how often friends of friends are also friends. 
We show that we can achieve speed-ups of up to four orders of magnitude with extremely small errors; see \Fig{gcc-time} and  \Fig{gcc-error}.
\begin{figure}[h]
  \centering
  \includegraphics[width=0.9\columnwidth]{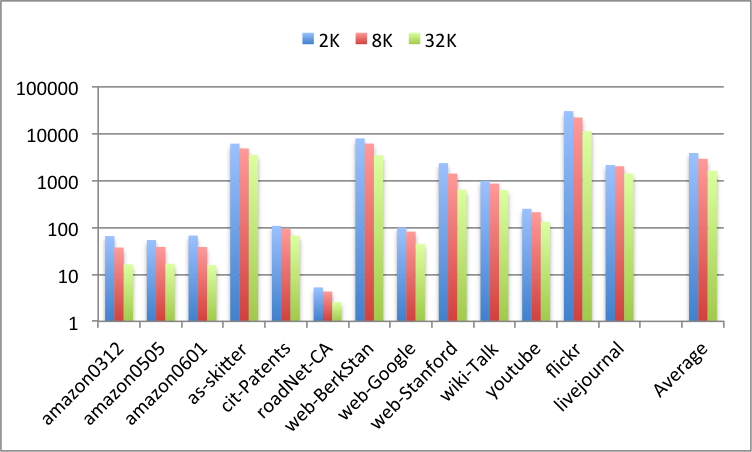}
  \caption{Speed-up over enumeration for global clustering coefficient computation with increasing numbers of wedge samples}
  \label{fig:gcc-time}
\end{figure}
\begin{figure}[h]
  \centering
  \includegraphics[width=0.9\columnwidth]{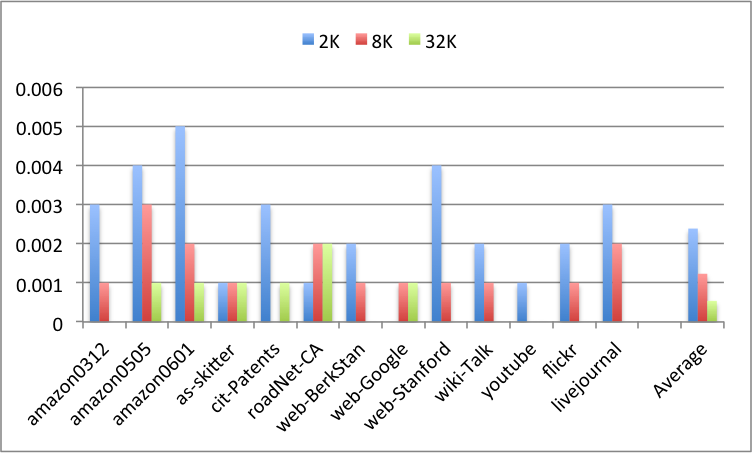}
  \caption{Absolute error in global clustering coefficient for increasing numbers of wedge samples}
  \label{fig:gcc-error}
\end{figure}

Our approach is not limited to clustering coefficients, however.
A per-vertex clustering coefficient, $C_v$,
is defined as the fraction of wedges centered at $v$ that participate in triangles.
The mean value of $C_v$ is called the \emph{local clustering coefficient}, $\lcc$.
A more nuanced view of triangles is given by the \emph{degree-wise clustering coefficients}.
This is a set of values $\{\dcc_d\}$ indexed by degree, where $\dcc_d$ is the average
clustering coefficient of degree $d$ vertices. In addition to degree distribution, many graphs
are characterized by plotting the clustering 
coefficients, i.e., $\dcc_d$ versus $d$. 
We summarize our notation and give formal definitions in \Tab{notation}.

\begin{table}[t]
\caption{Graph notation and clustering coefficients}
\label{tab:notation}
\begin{tabularx}{\linewidth}{@{}lX@{}}
\toprule
{$n$} &  number of vertices\\
{$n_d$} & number of vertices of degree $d$\\
{$m$} & number of edges\\
{$d_v$} & degree of vertex $v$\\
{$V_d$} & set of degree-$d$ vertices\\
{$W$} &  total number of wedges\\
{$W_v$} &  number of wedges centered at vertex $v$\\
{$T$} & total number of triangles \\
{$T_v$} & number of triangles incident to vertex $v$\\
{$T_d$} & number of triangles incident to a degree $d$ vertex\\
\bottomrule
\end{tabularx}
\begin{tabularx}{\linewidth}{@{}lX@{}}
\toprule
$\gcc = 3T/W$ & global clustering coefficient\\
$C_v = T_v/W_v$ & clustering coefficient of vertex $v$\\
$ \lcc = n^{-1} \sum_v C_v$ & local clustering coefficient\\
$ \dcc_d = n_d^{-1} \sum_{v \in V_d} C_v$ & degree-wise clustering coefficient\\
\bottomrule
\end{tabularx}
\end{table}

\subsection{Related Work}

There has been significant work on enumeration of all triangles~\cite{ChNi85,ScWa05,latapy08, BeFoNoPh11,ChCh11}.
Recent work by Cohen~\cite{Co09} and Suri and Vassilvitskii~\cite{SuVa11} give parallel implementations
of these algorithms. Arifuzzaman et al.~\cite{ArKhMa12} give a massively parallel algorithm for computing
clustering coefficients.  Enumeration algorithms however, can be very expensive, since  graphs even of moderate size (millions of vertices) can have
an extremely large number of triangles (see e.g., \Tab{prop}). 
Eigenvalue/trace based methods have been used by Tsourakakis~\cite{Ts08} and Avron~\cite{Av10} to compute
estimates of the total and per-degree number of triangles. However, computing eigenvalues (even just a few of them) is a compute-intensive task and  quickly becomes intractable on large graphs.

Most relevant to our work are sampling mechanisms.
Tsourakakis et al.~\cite{TsDrMi09} started the use of sparsification methods, the most important of which
is Doulion~\cite{TsKaMiFa09}.  This method sparsifies the graph by keeping each edge with probability $p$; counts  the triangles  in the sparsified graph; and multiplies this count by $p^{-3}$ to predict the  number of triangles in the original graph.  
Various theoretical analyses of this algorithm (and its variants) have been
proposed~\cite{KoMiPeTs10,TsKoMi11,PaTs12}. 
One of the main benefits of Doulion is that it reduces large graphs to smaller ones that can be loaded into memory. 
However, their estimate can suffer from high variance~\cite{YoKi11}.
Alternative sampling mechanisms have been proposed for streaming and semi-streaming algorithms \cite{BaKuSi02, JoGh05, BeBoCaGi08,BuFrLeMaSo06}.

Yet, all these fast sampling methods only estimate the number of triangles
and give no information about other triadic measures.

\subsection{Our contributions \label{sec:results}}
In this paper, 
we introduce the simple yet powerful technique of \emph{wedge sampling} for counting
triangles. Wedge sampling is really an algorithmic template, as opposed to a single algorithm, as various algorithms can
be obtained from the same basic structure.
Some of the salient features of this method are:
\begin{asparaitem}
\item \textbf{Versatility of wedge sampling:} We show how to use
  wedge sampling to approximate the various clustering coefficients:
  $\gcc$, $\lcc$, and $\{\dcc_d\}$. From these, we can estimate the
  numbers of triangles: $T$ and $\{T_d\}$. Moreover, our techniques
  can even be extended to find uniform random triangles, which is
  useful for computing triadic statistics.  Other sampling methods
  are usually geared towards only $T$ and $\gcc$.	
\item \textbf{Provably good results with precise bounds:} The
  mathematical analysis of this method is a direct consequence of
  standard Chernoff-Hoeffding bounds.  We obtain explicit
  time-error-accuracy tradeoffs. In other words, if we want an
  estimate for $\gcc$ with error at most 10\% with probability at
  least 99.9\% (say), we know we need only 380 wedge samples. This estimate
  is \emph{independent of the size of the graph}, though the
  preprocessing required by our method is linear in the number of
  edges (to obtain the degree distribution).
\item \textbf{Fast and scalable:} Our estimates converge rapidly, well within the theoretical bounds
  provided.  Although there is no other method that competes directly
  with wedge sampling, we compare with Doulion~\cite{TsKaMiFa09}.
  Our experimental studies show that  our  wedge sampling algorithm is far faster, when the variances of the two methods are similar (see \Tab{Dgcc} in the appendix).
  We do not compare to eigenvalue-based approaches since
 they are much more expensive for larger graphs.
\end{asparaitem}

\section{The wedge sampling method} \label{sec:wedge}

We present the general method of wedge sampling for estimating clustering coefficients. 
In later sections, we instantiate this for different algorithms.

We say a wedge is \emph{closed} if it is part of a
triangle; otherwise, we say the wedge is \emph{open}. Thus,
in \Fig{undirected-graph}, 
\circled{5}-\circled{4}-\circled{6} is an open wedge, while
\circled{3}-\circled{4}-\circled{5} is a closed wedge.
The middle vertex of a wedge is called its \emph{center}, i.e.,
wedges \circled{5}-\circled{4}-\circled{6} and
\circled{3}-\circled{4}-\circled{5} are centered at
\circled{4}.

\begin{figure}[htp]
  \centering
  \includegraphics[height=1in]{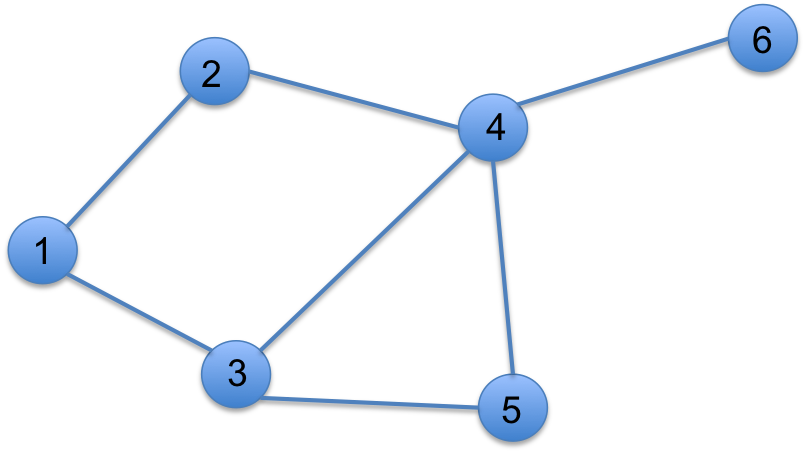}
  \caption{Example graph with 12 wedges and 1 triangle.}
  \label{fig:undirected-graph}
\end{figure}

Wedge sampling is best understood in terms of the following thought experiment.
Fix some distribution over wedges and
let $w$ be a random wedge. 
Let $X$ be the indicator random variable that is $1$
if $w$ is closed and $0$ otherwise.  Denote $\mu = \EX[X]$.

Suppose we wish to estimate $\mu$. We simply generate $k$
independent random wedges $w_1, w_2, \ldots, w_k$, with associated
random variables $X_1, X_2, \ldots, X_k$. Define
$\bar{X} = \frac{1}{k} \sum_{i \leq k} X_i$ as our estimate. The
Chernoff-Hoeffding bounds give guarantees on $\bar X$, as follows.
\begin{theorem}[Hoeffding \cite{Ho63}]
  \label{thm:Hoeffding}
  Let $X_1, X_2, \dots, X_k$ be independent random variables with $0
  \leq X_i \leq 1$ for all $i=1,\dots,k$.  Define $\bar X =
  \frac{1}{k} \sum_{i=1}^k X_i$. Let $\mu = \EX[\bar X]$. 
  Then for $\eps \in (0,1)$, we have
  \begin{displaymath}
    \Prob{ |\bar X - \mu | \geq \eps } \leq 2 \exp(-2 k \eps^2).
  \end{displaymath}
\end{theorem}
Hence, if we set $k = \cei{0.5 \eps^{-2}\ln(2/\delta)}$, then
$\Pr[|\bar{X} - \mu| > \eps] < \delta$. In other words,
with confidence at least $1-\delta$, 
the error in our estimate is at most $\eps$.

\Fig{level_curves} shows the number of samples needed
for different error rates. We show three different curves for
difference confidence levels. Increasing the confidence has minimal
impact on the number of samples. The number of samples is fairly low
for error rates of 0.1 or 0.01, but it increases with the inverse
square of the desired error. 
\begin{figure}[htp]
  \centering
  \includegraphics{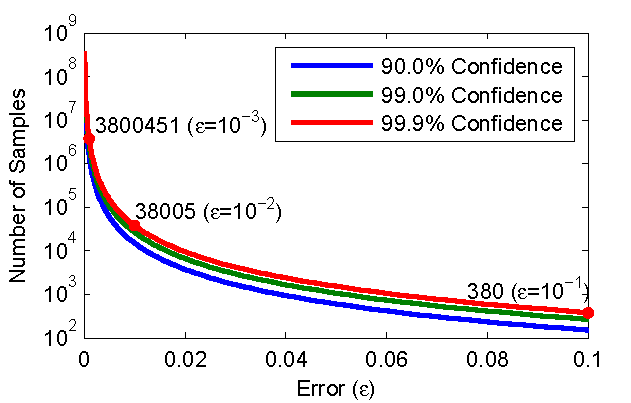}
  \caption{The number of samples needed for different error rates and
    different levels of confidence. A few data points at 99.9\%
    confidence are highlighted.}
  \label{fig:level_curves}
\end{figure}

\section{Computing the global clustering coefficient and the number of triangles} \label{sec:gcc}

\newcommand\cta[1]{\multicolumn{1}{|c|}{#1}}
\newcommand\ct[1]{\multicolumn{1}{c|}{#1}}
\begin{table*}[t]
 \caption{  Properties of the graphs }
\label{tab:prop}
  \centering\small
  \begin{tabular}{|>{\tt}r@{\,}|*{7}{r@{\,}|}}
    \hline
&&&&&&& Time\\
            \cta{Graph} &\ct{$n$}&\ct{$m$}&\ct{$W$}&\ct{$T$}&\ct {$\gcc$} &\ct{$\lcc$}& (secs) \\
    \hline
                          amazon0312 &   401K &  2350K &    69M &  3686K &	0.160 & 0.421 & 0.261	 \\
             amazon0505 &   410K &  2439K &    73M &  3951K &  	0.162 & 0.427& 0.269	 \\
             amazon0601 &   403K &  2443K &    72M &  3987K &  0.166 & 0.430& 0.268	 \\
             as-skitter &  1696K & 11095K & 16022M & 28770K & 0.005	& 0.296 & 90.897	 \\	
                        cit-Patents &  3775K & 16519K &   336M &  7515K & 0.067 & 0.092& 3.764	\\
            	             roadNet-CA &  1965K &  2767K &     6M &   121K & 0.060& 0.055 & 0.095 \\ 
          	           web-BerkStan &   685K &  6649K & 27983M & 64691K & 0.007& 0.634  & 54.777 \\
             web-Google &   876K &  4322K &   727M & 13392K & 0.055  & 0.624& 0.894	 \\
           web-Stanford &   282K &  1993K &  3944M & 11329K & 0.009 & 0.629& 6.987  \\
              wiki-Talk &  2394K &  4660K & 12594M &  9204K &  0.002 & 0.201&  20.572	  \\ 
              youtube & 1158K&  2990K& 1474M& 3057K& 0.006& 0.128 & 2.740  \\
              flickr & 1861K& 15555K  &14670M & 548659K &   0.112 &  0.375 & 567.160 \\
              livejournal & 5284K & 48710K& 7519M& 310877K& 0.124 & 0.345& 102.142  \\    \hline
 \end{tabular}
\end{table*}

We use the wedge sampling scheme to estimate the global clustering coefficient, $\gcc$. 
Consider the uniform distribution on wedges. We can interpret
$\EX[X]$ as the probability that a uniform random wedge is closed
or, alternately, the fraction of closed wedges.

To generate a uniform random wedge, note that the number of wedges
centered at vertex $v$ is $W_v = {d_v \choose 2}$ and $W = \sum_v
W_v$.  We set $p_v = W_v/W$ to get a distribution over the vertices.
Note that the probability of picking $v$ is proportional to the number
of wedges centered at $v$.
A uniform random wedge centered at $v$ can be generated by choosing two
random neighbors of $v$ (without replacement).

\begin{claim} \label{clm:random} 
  Suppose we choose vertex $v$ with probability $p_v$
  and take a uniform random pair of neighbors of $v$. 
  This generates a uniform random wedge.
\end{claim}

\begin{proof} Consider some wedge $w$ centered at vertex $v$. 
The probability that $v$ is selected is $p_v=W_v/W$.
The random pair has probability of $1/{d_v \choose 2} = 1/W_v$.
Hence, the net probability of $w$ is $1/W$.
\qed
\end{proof}

\begin{algorithm}
\caption{ \algkappa{} }
\label{alg:C}
\begin{algorithmic}[1]
\STATE	Compute $p_v$ for all vertices
\STATE	Select $k$ random vertices (with replacement) according to 
probability distribution defined by $\{p_v\}$. 
\STATE	 For each selected vertex $v$, choose a uniform random pair of
neighbors of $v$ to generate a wedge. 
\STATE	Output fraction of closed wedges as estimate of $\gcc$.
\end{algorithmic}
\end{algorithm}

\Alg{C} shows the randomized algorithm \algkappa{} for estimating $\gcc$ in a graph $G$.
Combining the bound of \Thm{Hoeffding} with \Clm{random}, we get the following theorem.
\begin{theorem} \label{thm:kappa} 
Set $k = \cei{0.5 \eps^{-2}\ln(2/\delta)}$. The algorithm \algkappa{}
outputs an estimate 
$\bar X$ for the clustering coefficient $\gcc$
such that $|\bar{X} - \gcc| < \eps$
with probability greater than $(1-\delta)$.
\end{theorem}
Note that the number
of samples required is \emph{independent of the graph size}, but
computing $p_v$ does depend on the number of edges, $m$. 

To get an estimate on $T$, the number of triangles, we output $\bar{X}
\cdot W/3$, which is guaranteed to be within $\pm \eps W/3$ of
$T$ with probability greater than $1-\delta$.

\subsection{Experimental results}
\label{sec:gcc-exp}

We  implemented our algorithms in {\tt C} and ran our experiments on a
computer equipped with a 2.3GHz Intel core i7 processor with 4~cores
and  256KB  L2 cache (per core), 8MB L3 cache, and an 8GB memory.  
We performed our experiments on 13 graphs  from
SNAP~\cite{Snap} and per private communication
with the authors of~\cite{MiMaGu07}.
In all cases, directionality is ignored, and repeated and  self-edges are omitted. 
The properties of these matrices are presented in \Tab{prop}, where
$n$, $m$, $W$, and $T$  are the numbers of vertices, edges, wedges,
and triangles, respectively. And $\gcc$ and $\lcc$ correspond to  the global and local clustering coefficients. 
The last column reports the times for the enumeration algorithm. Our enumeration algorithm is based on
the principles of \cite{ChNi85, ScWa05, Co09, SuVa11}, such that each
edge is assigned to its vertex with a smaller degree  (using the
vertex numbering as a tie-breaker), and  then vertices only check
wedges formed by edges assigned to them for closure. 

As  seen in \Fig{gcc-time}  wedge sampling works orders of magnitude faster than the enumeration algorithm. 
Detailed results on times can be seen in \Tab{gcc} in the appendix. 
The timing results show tremendous savings; for instance,   wedge
sampling only takes 0.015 seconds on \texttt{as-skitter} while full
enumeration takes 90 seconds.

\Fig{gcc-error} show the accuracy of  the wedge sampling algorithm. Again detailed results on times can be seen in \Tab{gcc} in the appendix. 
At 99.9\% confidence ($\delta=0.001$), the upper bound on the error we expect for 2K, 8K,
and 32K samples is .043, .022, and .011, respectively. Most of the
errors are much less than the bounds would suggest. For instance, the
maximum error for 2K samples is .007, much less than that 0.43 given
by the upper bound.
\Fig{amazon0505_cc} shows the fast convergence of the clustering
coefficient estimate (for the graph {\tt amazon0505}) as the number of samples increases. The dashed
line shows the error bars at 99.9\% confidence.  
In all our experiments, the real error is
always much smaller than  what is  indicated by  \Thm{Hoeffding}. 
\begin{figure}[thp]
  \centering
  \includegraphics{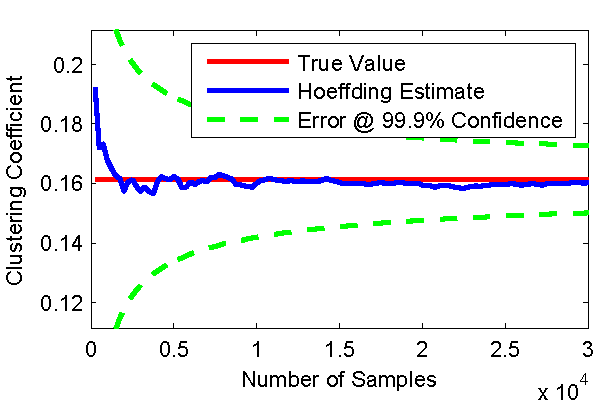}
  \caption{Convergence of clustering coefficient estimate as the
    number of samples increases for  \texttt{amazon0505}.}
  \label{fig:amazon0505_cc}
\end{figure}

\section{Computing the local clustering coefficient} \label{sec:lcc}

We now demonstrate how a small change to the underlying distribution on wedges
allows us to compute the local clustering coefficient,
$\lcc$. \Alg{lcc} shows
the procedure \alglcc{}.
The only difference between  \alglcc{} and \algkappa{} is in picking
random vertex centers. Vertices are picked uniformly instead of from the distribution $\{p_v\}$. 
\medskip
\begin{algorithm}
\caption{ \alglcc{} \label{alg:lcc}}
\begin{algorithmic}[1]
\STATE Pick $k$ uniform random vertices.
\STATE For each selected vertex $v$, choose a uniform random pair of neighbors of $v$ to generate a wedge.
\STATE Output the fraction of closed wedges as an estimate for $\lcc$.
\end{algorithmic}
\end{algorithm}

\begin{theorem} \label{thm:lcc} 
Set $k = \cei{0.5 \eps^{-2}\ln(2/\delta)}$. The algorithm \alglcc{}
outputs an estimate 
$\bar X$ for the clustering coefficient $\lcc$ such that 
$|\bar{X} - \lcc| < \eps$
with probability greater than $(1-\delta)$.
\end{theorem}

\begin{proof} 
Let us consider a single sampled wedge $w$, and let $X(w)$
be the indicator random variable for the wedge being closed. Let $\cV$
be the uniform distribution on wedges. For any vertex $v$, let $\cN_v$
be the uniform distribution on pairs of neighbors of $v$.
Observe that 
$$ {\EX[X]} = \Pr_{v \sim \cV}[\Pr_{(u,u') \sim \cN_v}[\textrm{wedge $\{(u,v),(u',v)\}$ is closed}]] $$
We will show that this is exactly $\lcc$.
\begin{eqnarray*}
\lcc & = & n^{-1} \sum_v C_v = \EX_{v \sim \cV} [C_v] \\
& = & \EX_{v \sim \cV} [\textrm{frac. of closed wedges centered at $v$}] \\
& = & \EX_{v \sim \cV} [\EX_{(u,u') \sim \cN_v}[X(\{(u,v),(u',v)\})]] \\
& = & \Pr_{v \sim \cV}[\Pr_{(u,u') \sim \cN_v}[\textrm{$\{(u,v),(u',v)\}$ is closed}]] \\
& = & \EX[X]
\end{eqnarray*}
For a single sample, the probability that the wedge is closed is exactly $\lcc$.
The bound of \Thm{Hoeffding} completes the proof.
\qed
\end{proof}

\Fig{lcc-error} and \Fig{lcc-times} present  the results of our experiments for computing the local clustering coefficients. 
More detailed results can be found \Tab{lcc} in the appendix. 
Experimental setup and the notation are  the same as in \Sec{gcc-exp}. 
The results again show that wedge sampling provides accurate
estimations with tremendous improvements in runtime. In this case, we
come closer to the theoretical error bounds. For instance, the largest
different in the 2K sample case is 0.017, which is much closer to the
theoretical error bound of 0.043.

\begin{figure}[ht]
  \centering
  \includegraphics[width=.45\textwidth]{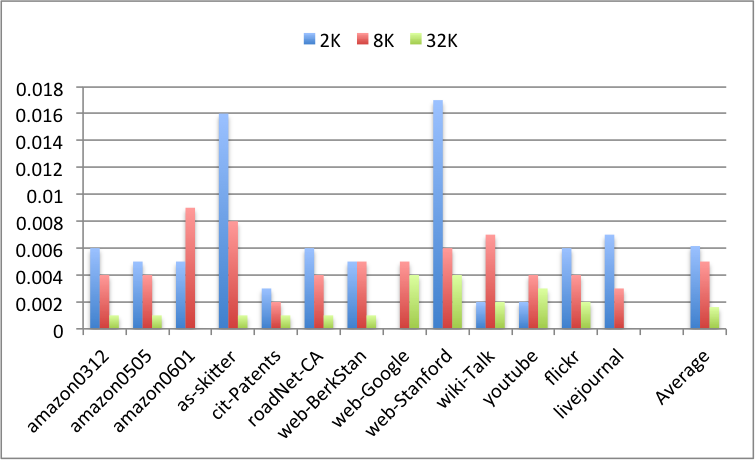}
  \caption{ Absolute error in local clustering coefficient for increasing numbers of wedge samples}
  \label{fig:lcc-error}
\end{figure}

\begin{figure}[ht]
  \centering
  \includegraphics[width=.45\textwidth]{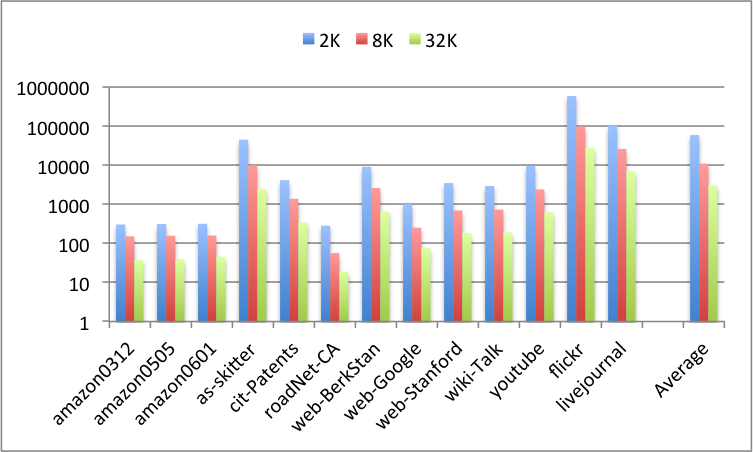}
  \caption{ Speed-up in local clustering coefficient computation time for increasing numbers of wedge samples}
  \label{fig:lcc-times}
\end{figure}

\section{Computing degree-wise clustering\\ coefficients and triangle estimates} \label{section:deg-cc}

\begin{figure*}[tb]
\centering

    \subfloat{
    \includegraphics[width=2in,trim=0 0 0 0]{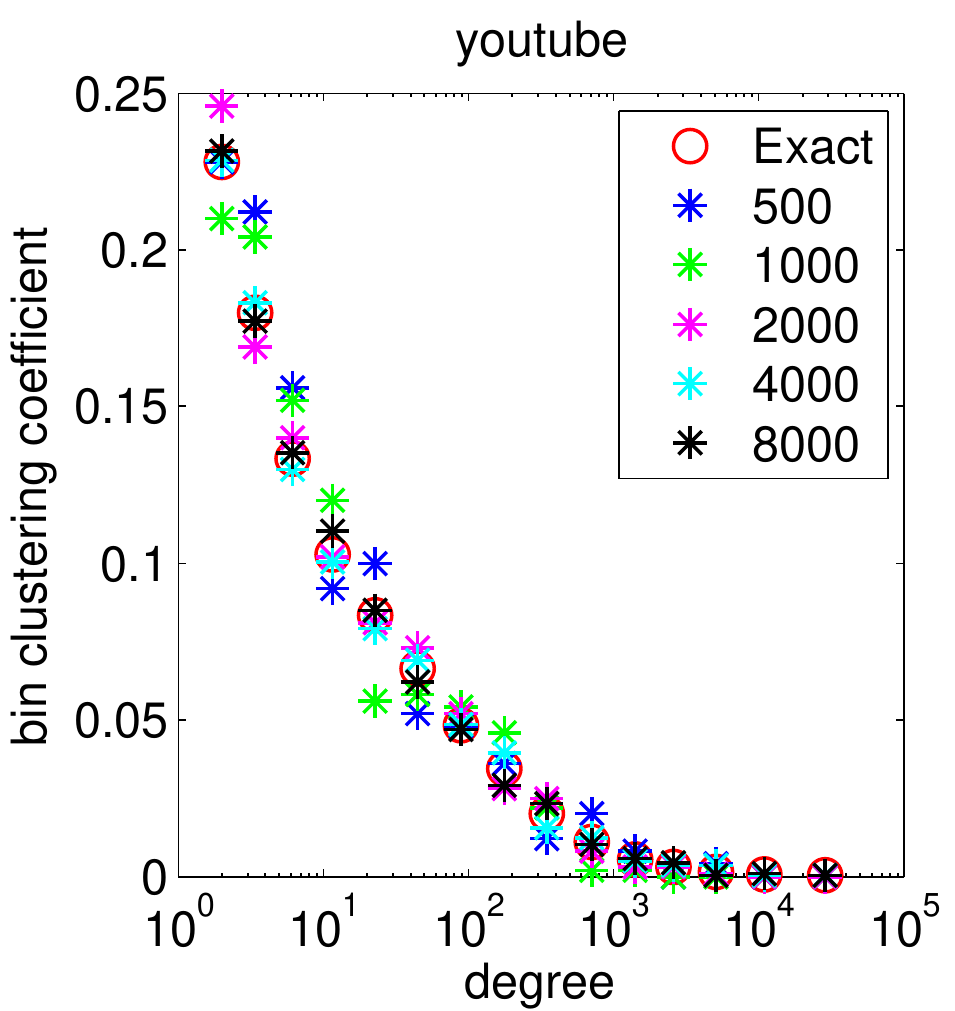}
  }
    \subfloat{
    \includegraphics[width=2in,trim=0 0 0 0]{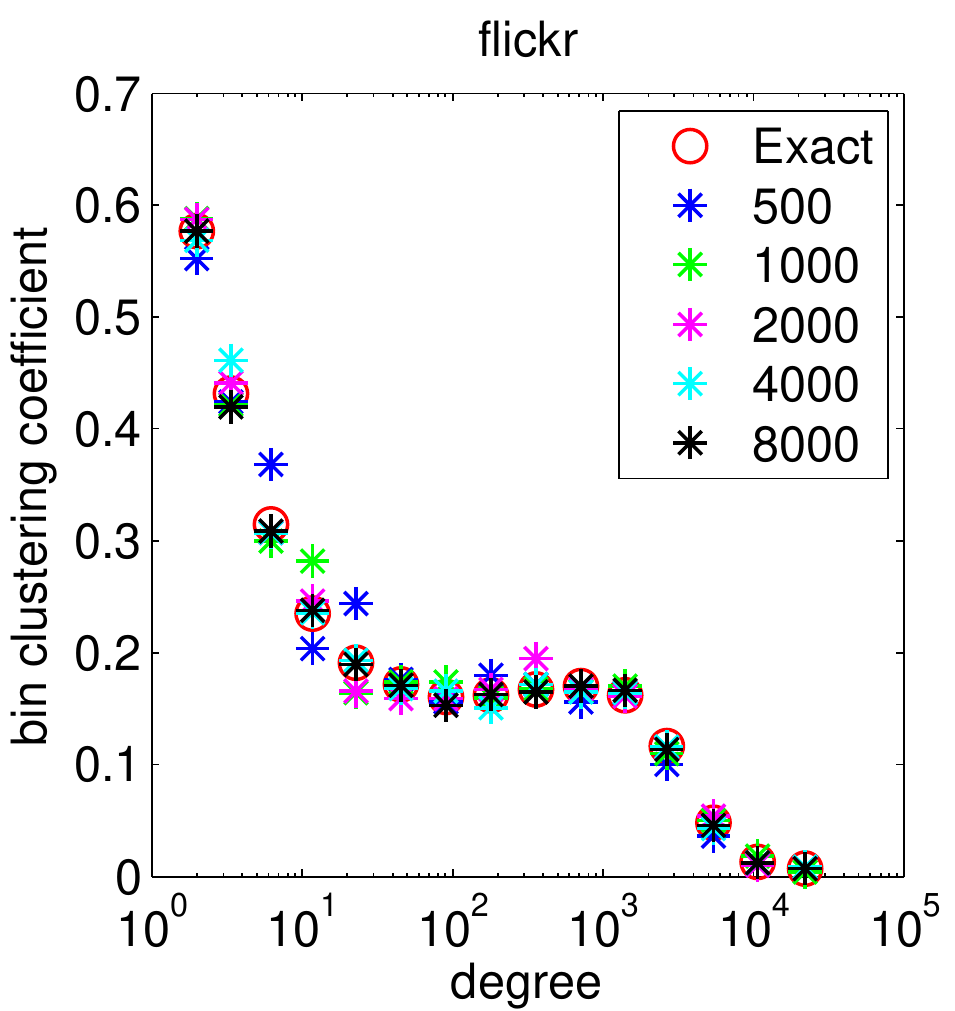}
  }
   \subfloat{
    \includegraphics[width=2in,trim=0 0 0 0]{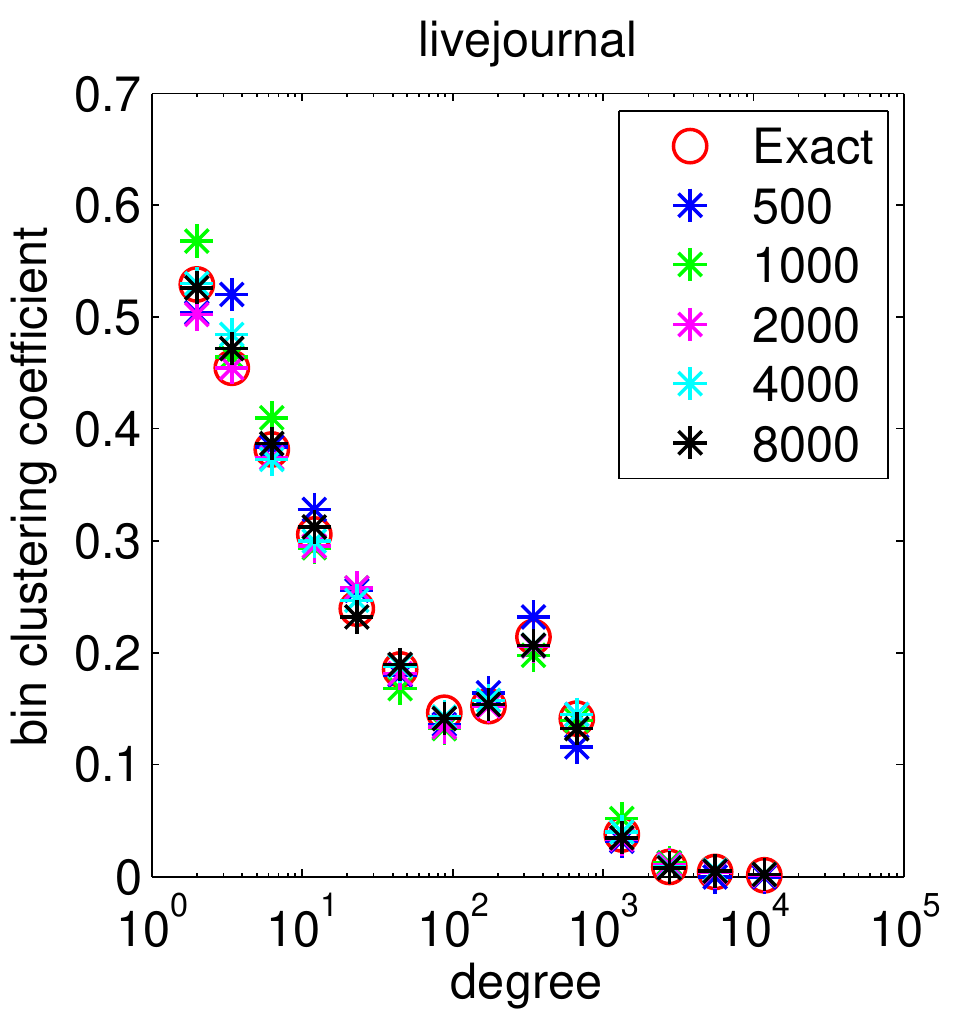}
  }
  
\caption{Computing degree-wise clustering coefficients using wedge
  sampling}
\label{fig:ccd}
\end{figure*}

We demonstrate the power of wedge sampling by estimating the
degree-wise clustering coefficients
$\{\dcc_d\}$. We also provide a sampling algorithm to estimate $T_d$, the number of triangles
incident to degree-$d$ vertices. \Alg{dcc} shows procedure \algdcc.

\begin{algorithm}
\caption{ \algdcc \label{alg:dcc}}
\begin{algorithmic}[1]
\STATE	 Pick $k$ uniform random vertices of degree $d$.
\STATE	 For each selected vertex $v$, choose a uniform random pair of neighbors of $v$ to generate a wedge.
\STATE	 Output the fraction of closed wedges as an estimate for $\dcc_d$.
\end{algorithmic}
\end{algorithm}

\begin{theorem} \label{thm:dcc} 
Set $k = \cei{0.5 \eps^{-2}\ln(2/\delta)}$. The algorithm \algdcc{} outputs an estimate
$\bar X$ for the clustering coefficient $\dcc_d$ such that
$|\bar{X} - \dcc_d| < \eps$
with probability greater than $(1-\delta)$.
\end{theorem}

\begin{proof} The proof of the following is similar to that
of \Thm{lcc}. Since $\dcc_d$ is the average clustering coefficient of a 
degree $d$ vertex, we can apply the same arguments as in \Thm{lcc}.
\qed
\end{proof}

By modifying the template given in \Sec{wedge}, we can also get estimates
for $T_d$. Now, instead of simply counting the fraction of closed wedges,
we take a weighted sum. \Alg{trid} describes the procedure \algtrid.
We let
$W_d = n_d \cdot {d \choose 2}$ denote the total number of wedges centered at degree $d$ vertices.

\begin{algorithm}
\caption{ \algtrid{} \label{alg:trid}}
\begin{algorithmic}[1]
\STATE Pick $k$ uniform random vertices of degree $d$. 
\STATE For each selected vertex $v$, choose a uniform random pair of neighbors of $v$ to generate a wedge.
\STATE For each wedge $w_i$ generated, let $Y_i$ be the associated random
variable such that
\begin{displaymath}
  Y_i =
  \begin{cases}
    0 & \text{if $w$ is open},\\
    \frac{1}{3} & \text{if $w$ is closed and has 3 vertices in $V_d$},\\
    \frac{1}{2} & \text{if $w$ is closed and has 2 vertices in $V_d$},\\
    1 & \text{if $w$ is closed and has 1 vertex in $V_d$}.
  \end{cases}
\end{displaymath}
\STATE $\bar Y = \frac{1}{k} \sum_i Y_i$.
\STATE Output $W_d \cdot \bar{Y}$ as the estimate for $T_d$.
\end{algorithmic}
\end{algorithm}
\medskip

\begin{theorem} \label{thm:dcc2} 
Set $k = \cei{0.5 \eps^{-2}\ln(2/\delta)}$. 
The algorithm \algtrid outputs an estimate
$W_d \cdot \bar Y$ for the $T_d$ with the following guarantee:
$|W_d \cdot \bar{Y} - T_d| < \eps W_d$
with probability greater than $1-\delta$.
\end{theorem}

\begin{proof} For a single sampled wedge $w_i$, we define $Y_i$. We will
argue that the expected value of $\EX[Y]$ is exactly $T_d/W_d$
below. Once we have that, an application
of the Hoeffding bound of \Thm{Hoeffding} shows that $|\bar Y - T_d/W_d| < \eps$ with
probability greater than $1-\delta$. Multiplying this inequality by
$W_d$, we get $|W_d \cdot \bar{Y} - T_d| < \eps W_d$,
completing the proof.

To show $\EX[Y] = T_d/W_d$, partition the set $S$ of wedges centered on degree $d$ vertices into four
sets $S_0, S_1, S_2, S_3$. The set $S_i$ ($i \neq 0$) contains all
closed wedges containing
exactly $i$ degree-$d$ vertices. The remaining open wedges
go into $S_0$. For a sampled wedge $w$, if $w \in S_i, i \neq 0$, then
$Y_i = 1/i$. If $w \in S_0$, then nothing is added. The wedge $w$ is 
a uniform random wedge from those centered on degree-$d$ vertices. Hence,
$\EX[Y] = |S|^{-1}(|S_1| + |S_2|/2 + |S_3|/3)$. 

Now partition the set of triangles involving degree $d$ vertices into
three sets $S'_1, S'_2, S'_3$, where $S'_i$ is the set of triangles
with $i$ degree $d$ vertices. Observe that $|S_i| = i|S'_i|$. If a triangle
has $i$ vertices of degree $d$, then there are exactly $i$ wedges
centered in degree $d$ vertices (in that triangle). 
So, $|S_1| + |S_2|/2 + |S_3|/3 = $ $|S'_1| + |S'_2| + |S'_3|$ $= T_d$. Therefore,
$\EX[Y] = T_d/W_d$.
\qed
\end{proof}

\subsection{Computing the clustering coefficient for bins of vertices}
Algorithms in the previous section present how to compute the
clustering coefficient of vertices of a given degree.  In practice, it
may be sufficient to compute clustering coefficients over bins of
degrees.  Wedge sampling algorithms can still be extended for bins of
degrees by a small adjustment of the sampling procedure.
Within each bin, we weight each vertex according to the number of
wedges it produces. This 
guarantees that each wedge in the bin is equally likely to be selected.  

\subsection{Experimental results} 
\Fig{ccd} shows results on three graphs for clustering coefficients;
the remaining figures are shown in the \Fig{ccd-app} in the appendix. 
The data is grouped in logarithmic
bins of degrees, i.e., $\set{2}, \set{3,4}, \set{5,6,7,8}, \dots$.
In other words, $2^{i-1} <d_v \leq 2^i$ form
the $i$-th bin.  
We show the estimates with increasing number of samples. At 8K
samples, the error is expected to be less than 0.02, which is
apparent in the plots.
Observe that even 500 samples yields a reasonable estimate in terms of
the differences by degree.

\Fig{ccdtimes} shows the time to calculate all $C_d$ values, showing a
tremendous savings in runtime as a result of using wedge sampling. In
this figure, runtimes are normalized with respect to the runtimes of
full enumeration.  As the figure shows, wedge sampling takes only a
tiny fraction of the time of full enumeration especially for large
graphs.

\begin{figure}[htp]
  \centering
  \includegraphics[width= 0.45\textwidth]{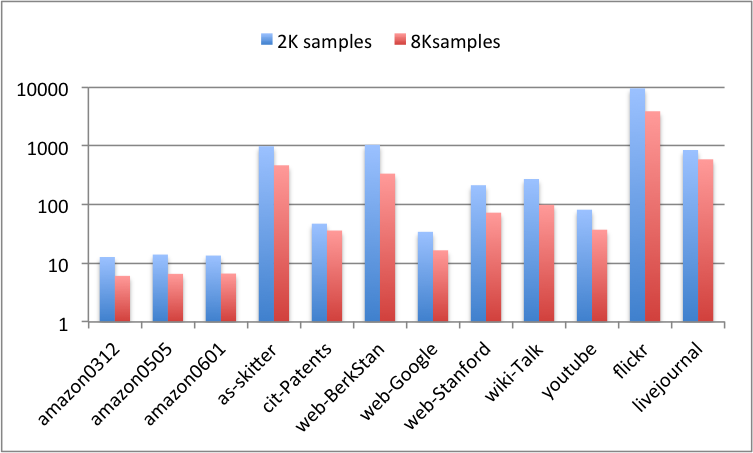}
  \caption{Speed-up in degree-wise clustering coefficient computation time for increasing numbers of wedge samples}
  \label{fig:ccdtimes}
\end{figure}

\section{Generating a uniform sample of the triangles}
While most triadic measures focus on  the number of triangles  and their distribution,  the  triangles themselves,  not only their count,  can reveal a lot of information about the graph. The authors' recent work~\cite{DuPiKo12} has looked at the relations among the degrees of the vertices of triangles. In these experiments, a full enumeration of the triangles was used, which limited the sizes of the graphs we could use.  To avoid this computational burden, 
a uniform sampling of the triangles can be used. 

{\em Wedge sampling provides a uniform sampling the triangles of a graph}.  Consider the uniform wedge sampling of \Alg{C}. Some of these wedges will be closed, so we 
generate a random set of triangles. Each such triangle is an independent, uniform random triangle. This is because wedges are chosen from the uniform distribution,
and every triangles contains exactly $3$ closed wedges.
\Fig{trisample} presents the results using triangle sampling to estimate the percentage of triangles where the maximum  to minimum degree ratio is $\geq10$, which is  motivated by an experiment in~\cite{DuPiKo12}.   The figure shows that accurate results can be achieved by using only 500 triangles and that  wedge sampling provides an unbiased selection of triangles. 
The expected number of wedges to be sampled to generate $T_s$  triangles  is $3T_s/\gcc$, which means the method  will be effective  unless the clustering coefficient is extremely small.

\begin{figure}[htp]
  \centering
  \includegraphics[width= 0.45\textwidth]{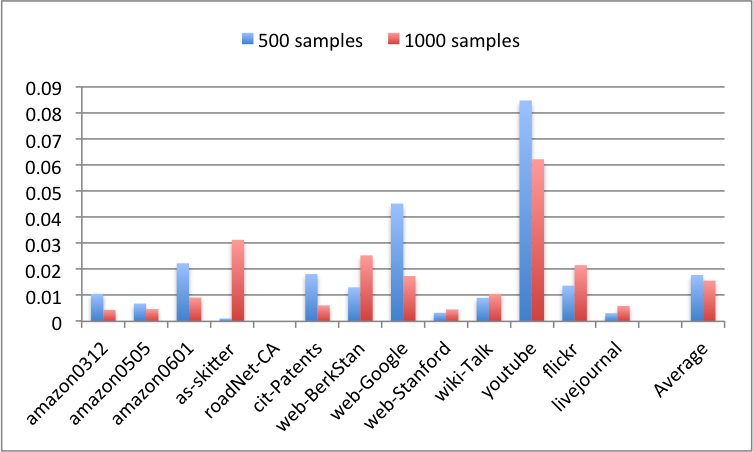}
  \caption{Error in computing  percentage of triangles where the maximum  to minimum degree ratio is $\geq10$ for increasing number of samples.  }
  \label{fig:trisample}
\end{figure}

\section {Comparison to Doulion}
\label{sec:comp}

Doulion~\cite{TsKaMiFa09} is an alternative sampling mechanism for estimating the number
of triangles in a graph. It has a single parameter $p$. 
Each edge is chosen independently at random with probability $p$, leading to a subgraph
of expected size $pm$ ($m$ is the total number of edges). We count the number
of triangles $T'$ in this subgraph, and estimate the total number of triangles
by $T'/p^3$. It is not hard to verify that this is correct in expectation, but
bounding the variance requires a lot more work. Some concentration
bounds for this estimate are known \cite{TsKoMi11,PaTs12}, but they depend
on the maximum number of triangles incident to an edge in the graph.
So they do not have the direct form of \Thm{kappa}. Some bad examples for
Doulion have been observed~\cite{YoKi11}. 
\begin{figure}[htp]
  \centering
  \includegraphics[width= 0.45\textwidth]{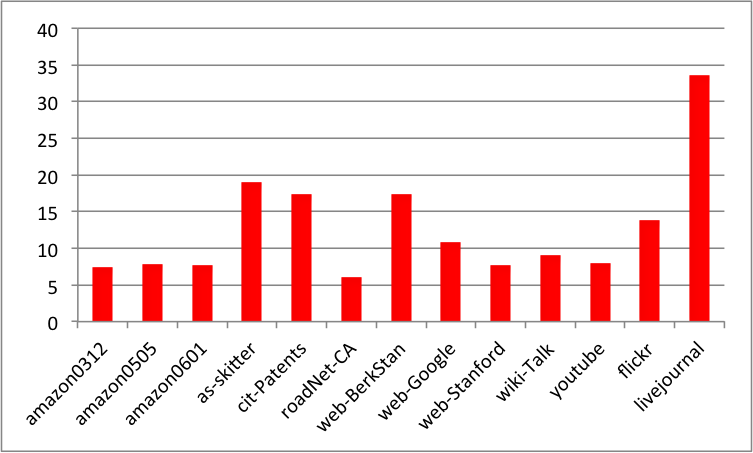}
  \caption{Speedups of wedge sampling with 32K samples over Doulion with $p=1/25$.}
  \label{fig:Douliontimes}
\end{figure}
Doulion is extremely elegant and simple, and leads to an overall reduction in graph size
(so a massive graph can now fit into memory).
Common values of $p$ used are $p=1/10$ and $p=1/25$.

We show that wedge sampling performs at least as good as Doulion in terms of accuracy, and has better runtimes.
We run wedge sampling with 32K samples. We start with setting the Doulion parameter $p = 1/25$, which has 
been used in the literature. Note that the size of the Doulion sample is $m/25$, which is much
larger than 32K. (For {\tt amazon0312}, one of the smaller graphs we consider, $m/25 \approx 90K$.)	
We run each algorithm 100 times. The (average) runtimes are compared in \Fig{Douliontimes}, where
we see that wedge sampling is always competitive. The accuracy of the estimate
is comparable for both algorithms. In \Tab{Dgcc} in the appendix, we present the minimum, maximum, and standard deviation for the 100 runs.
This also shows the good convergence properties of both algorithms, since even the minimum and maximum values are fairly close to the true clustering coefficient.
We also try $p=1/50$ and note that wedge sampling with 32K samples is still faster in terms of runtime while offering somewhat better accuracy.
By setting $p=1/100$, the sample size of Doulion becomes comparable to 32K, but Doulion is quite inaccurate (the range
between the maximum and the minimum is large).

In \Tab{Dlcc} of the appendix, we also compare results for the local clustering coefficient.

\section{Significance and Impact} 
We have proposed a series of wedge-based algorithms for computing
various triadic measures on graphs. Our algorithms come with
theoretical guarantees in the form of specific error and confidence
bounds.  We want to stress that the number of samples required to
achieve a specified error and confidence bound is independent of the
graph size.  For instance, 38,000 samples guarantees and error less
than 0.1\% with 99.9\% confidence \emph{for any graph}, which gives
our algorithms an incredible potential for scalability. The limiting
factors have to do with determining the sampling proportions; for instance, we
need to know the degree of each vertex and the overall degree
distribution to compute the global clustering coefficient.

The flexibility of wedge sampling along with the hard error bounds
essentially redefines what is feasible in terms of graph analysis. The
very expensive computation of clustering coefficient is now much
faster (enabling more near-real-time analysis of networks) and we can
consider much larger graphs than before. In an extension of this work,
we are pursuing a MapReduce implementation of this method that scales
to O(100M) nodes and O(1B) edges, needing only a few minutes of
compute time (based on preliminary results).

With triadic analysis no long being a computational burden, we can
extend our horizons into new territories and look at directed
triangles, attributed triangles (e.g., we might compare the clustering
coefficient for ``male'' and ``female'' nodes in a social network),
evolution of triadic connections, higher-order structures such a
4-cycles and 4-cliques, and so on.

\clearpage
\appendix 

\section{Appendix}

Various detailed results are shown in the appendix.

We present the runtime results for global clustering coefficient computations in \Tab{Dgcc}. As mentioned earlier, we perform 100 runs of each
algorithm and show the minimum, maximum, and standard deviations of the output estimates. We also show the relative speedup of wedge
sampling over Doulion with $p=1/50$.

We also used Doulion to compute the local clustering coefficient. For this purpose, we predicted the number of triangles incident to a vertex by counting  the triangles  in the sparsified graph and then dividing this number by $p^3$. The results of our experiments are presented in \Tab{Dlcc}, which shows that Doulion fails in accuracy even for $p = 1/10$. 
This is because  local clustering coefficient  is a finer level statistic,  which becomes a challenge for Doulion. Wedge sampling on the other hand, keeps its accurate estimations with low variance.

\begin{table*}[p]	
 \caption{  Estimating the global clustering coefficient }
\label{tab:gcc}
  \centering\small
  \begin{tabular}{|>{\tt}r@{\,}|*{12}{r@{\,}|}}
    \cline{7-13}
    \multicolumn{6}{c|}{} 
    & \multicolumn{3}{c|}{Wedge Sampling}
    & \multicolumn{4}{c|}{Time (sec)} \\
    \hline
            \cta{Graph} &\ct{$n$}&\ct{$m$}&\ct{$W$}&\ct{$T$}&\ct {$\gcc$} &\ct{2K}&\ct{8K}&\ct{32K} &\ct{E}&\ct{2K}&\ct{8K}&\ct{32K} \\
    \hline
                          amazon0312 &   401K &  2350K &    69M &  3686K &	0.160&	0.163&	0.161&	0.160  & 0.261	& 0.004& 	0.007& 	0.016\\ 
             amazon0505 &   410K &  2439K &    73M &  3951K &  	0.162&	0.158&	0.165&	0.163  & 0.269	&  0.005& 	0.007& 	0.016\\
             amazon0601 &   403K &  2443K &    72M &  3987K &  0.166&	0.161&	0.164&	0.167 & 0.268	& 0.004 &	0.007& 	0.017\\
             as-skitter &  1696K & 11095K & 16022M & 28770K & 0.005	&	0.006&	0.006&	0.006 & 90.897		& 0.015& 	0.019& 	0.026\\
                        cit-Patents &  3775K & 16519K &   336M &  7515K & 0.067&	0.064&	0.067&	0.068 & 3.764	& 0.035& 	0.040& 	0.056\\
            	             roadNet-CA &  1965K &  2767K &     6M &   121K & 0.060&	0.061&	0.058&	0.058 & 0.095& 	0.018& 	0.022& 	0.037\\
          	           web-BerkStan &   685K &  6649K & 27983M & 64691K & 0.007&	0.005&	0.006&	0.007  & 54.777&  	0.007& 	0.009& 	0.016\\
             web-Google &   876K &  4322K &   727M & 13392K & 0.055&	0.055&	0.054&	0.056 & 0.894	& 	0.009& 	0.011& 	0.020 \\
           web-Stanford &   282K &  1993K &  3944M & 11329K & 0.009&	0.013&	0.008&	0.009 & 6.987 &	0.003	& 0.005& 	0.011\\
              wiki-Talk &  2394K &  4660K & 12594M &  9204K &  0.002	&0.004&	0.003&	0.002 & 20.572	  & 0.021& 	0.024	& 0.033 \\
              youtube & 1158K&  2990K& 1474M& 3057K& 0.006& 0.005& 0.006&0.006& 2.740 &0.011 & 0.013& 0.021\\
              flickr & 1861K& 15555K  &14670M & 548659K &   0.112 &0.110 & 0.113 &0.112&   567.160&0.019 &0.026 & 0.051\\
              livejournal & 5284K & 48710K& 7519M& 310877K& 0.124 &0.127 &0.126 &0.124 &102.142 & 0.048& 0.051& 0.073\\
    \hline
 \end{tabular}
\end{table*}

\begin{table*}[p]
\caption{Estimating the local clustering coefficients}  
\label{tab:lcc} 
  \centering\small
  \begin{tabular}{|>{\tt}r@{\,}|*{9}{r@{\,}|}}
    \hline
    \cline{3-9}
    \multicolumn{1}{|c|}{} 
    &\multicolumn{1}{c|}{} 
    & \multicolumn{3}{c|}{Estimate}
    & \multicolumn{4}{c|}{Time (sec)} \\    \cline{3-9}
            \cta{Graph} & \ct {$\lcc$} &\ct{2K}&\ct{8K}&\ct{32K} &\ct{E}&\ct{2K}&\ct{8K}&\ct{32K} \\\hline
amazon0312 &	0.421	&	0.427&	0.417&	0.420 & 0.301	& 0.001	& 0.002	& 0.008\\
amazon0505&	0.427	&	0.422&	0.423&	0.426 & 0.310		& 0.001& 	0.002& 	0.008\\
amazon0601	&0.430	&	0.435&	0.421&	0.430 & 0.314	& 0.001& 	0.002	& 0.007\\
as-skitter	&0.296 &	0.280&	0.288&	0.297 & 88.290		& 0.002	& 0.009& 	0.036\\
cit-Patents&	0.092&		0.089&	0.094&	0.091&  4.081 	& 0.001	& 0.003	& 0.012\\
roadNet-CA&	0.055&		0.049&	0.059&	0.054 & 0.112	& 0.000	& 0.002& 	0.006\\
web-BerkStan	&0.634 &	0.629&	0.639&	0.633&  53.892	& 0.006	& 0.021& 	0.085\\
web-Google	&0.624	&0.624	&0.619&	0.628 & 0.996	& 	0.001	& 0.004	& 0.013 \\
web-Stanford&	0.629&		0.612&	0.635&	0.633 & 6.868	& 0.002& 	0.010	& 0.038\\
wiki-Talk	&0.201&	0.199&	0.194&	0.199 & 20.254		& 0.007	& 0.028& 	0.108 \\
 youtube & 0.128& 0.130&0.132 & 0.131&18.948 & 0.002&0.008 & 0.031\\
              flickr & 0.375 & 0.369  &0.371 &0.377 & 575.493 &0.001 & 0.006& 0.021\\ 
              livejournal & 0.345 & 0.338 &0.348 & 0.345& 102.142 & 0.001 & 0.004&0.015 \\\hline
\end{tabular}
\end{table*}

\begin{figure*}[p]
\centering
 
 \subfloat{\label{fig:ccd:amazon0505}
    \includegraphics[width=2in,trim=0 0 0 0]{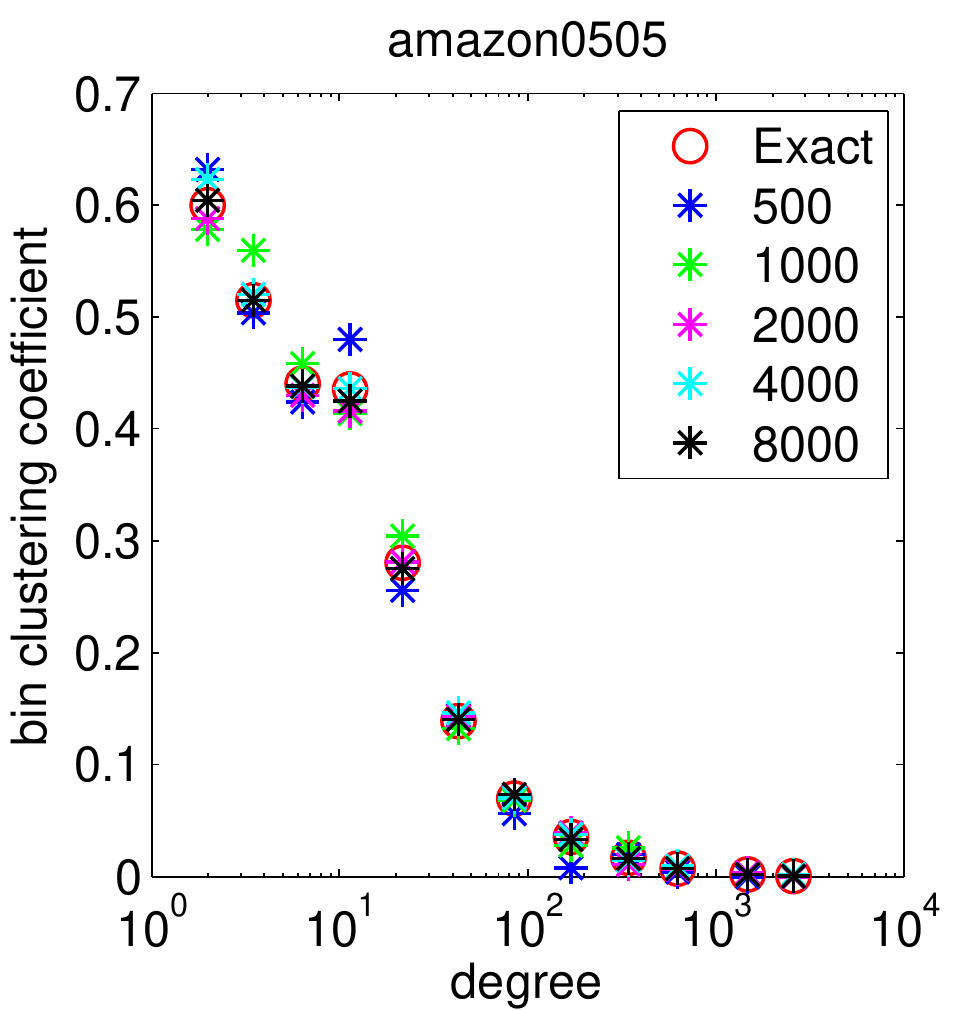}
  }
    \subfloat{\label{fig:ccd:amazon0601}
    \includegraphics[width=2in,trim=0 0 0 0]{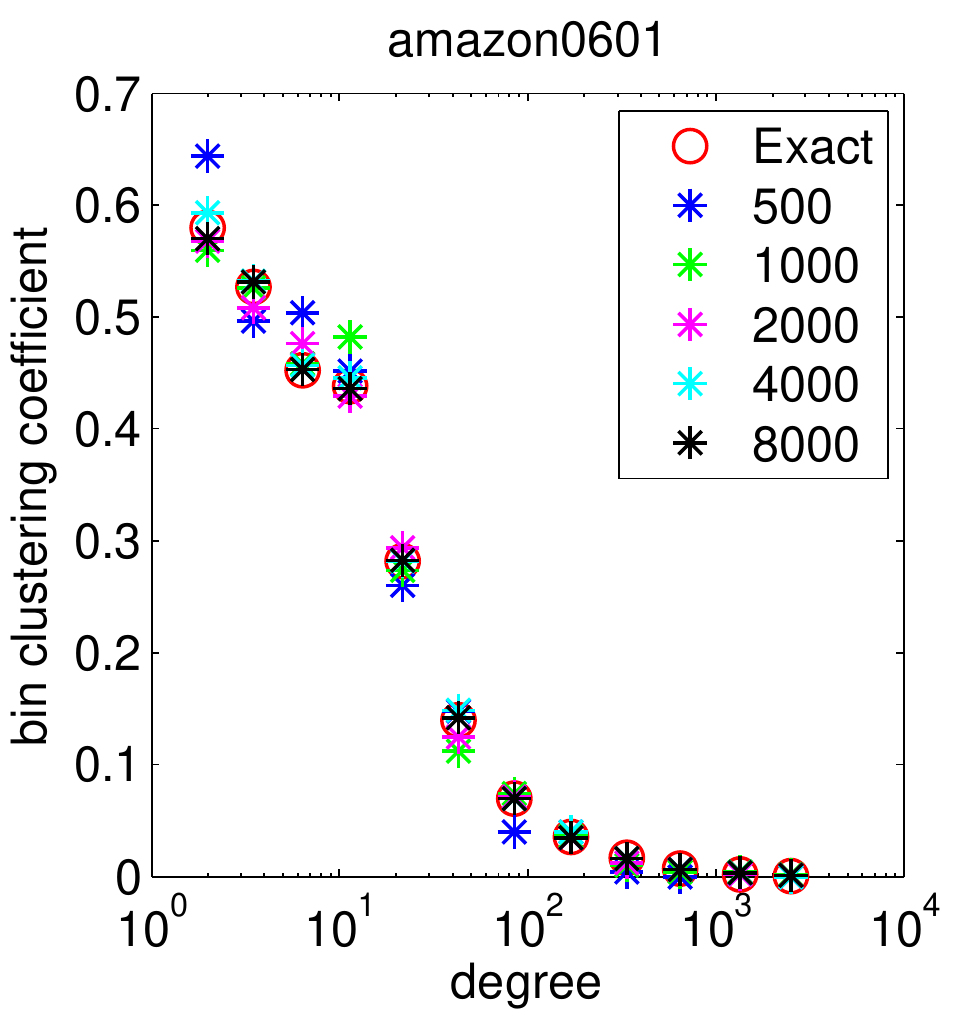}
  }
  \subfloat{\label{fig:ccd:as-skitter}
    \includegraphics[width=2in,trim=0 0 0 0]{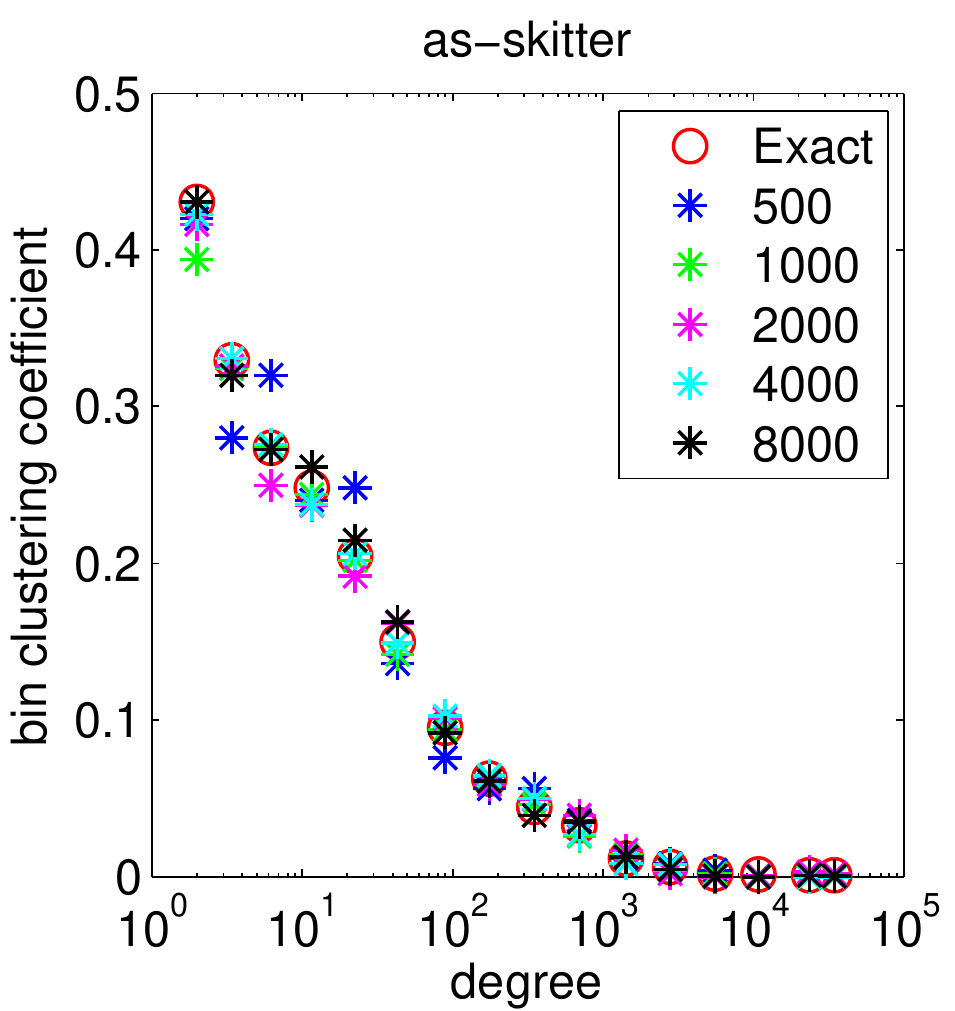}
  }
  \\
  \subfloat{\label{fig:ccd:cit-Patents}
    \includegraphics[width=2in,trim=0 0 0 0]{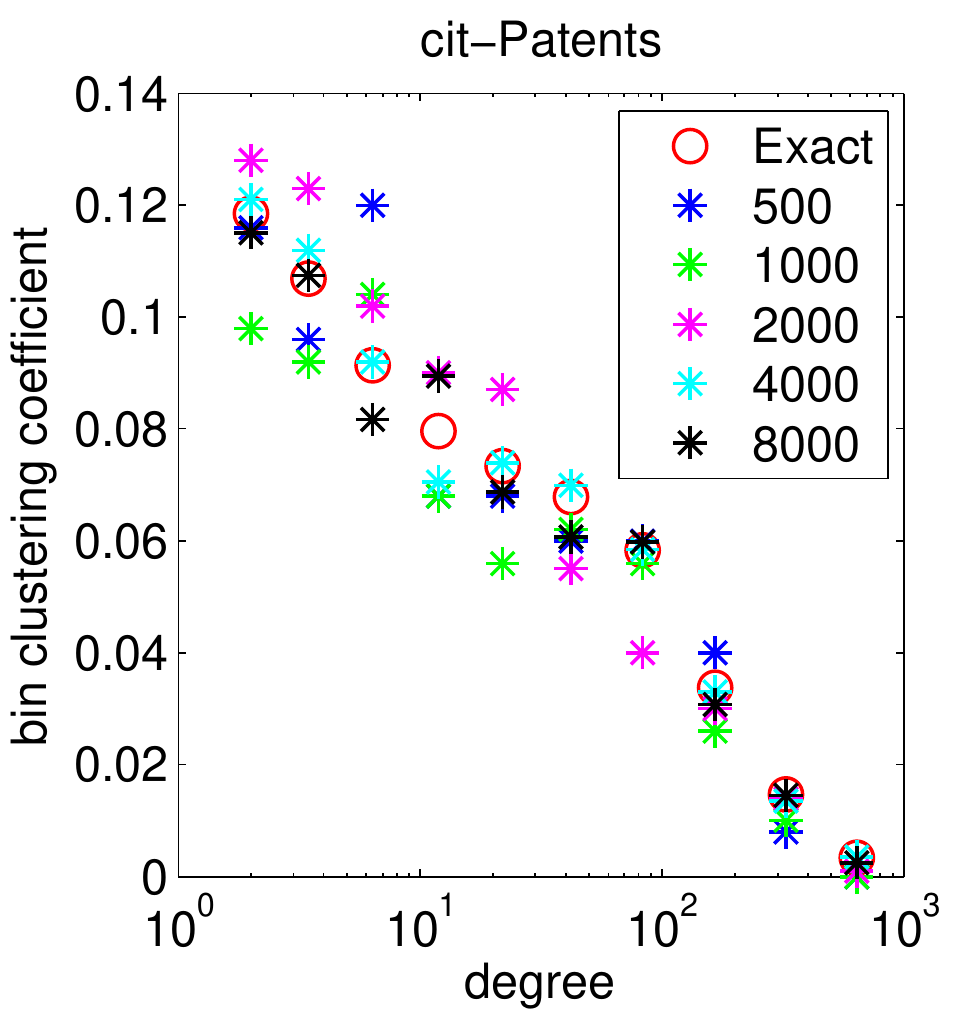}
  } 
  \subfloat{\label{fig:ccd:web-BerkStan}
    \includegraphics[width=2in,trim=0 0 0 0]{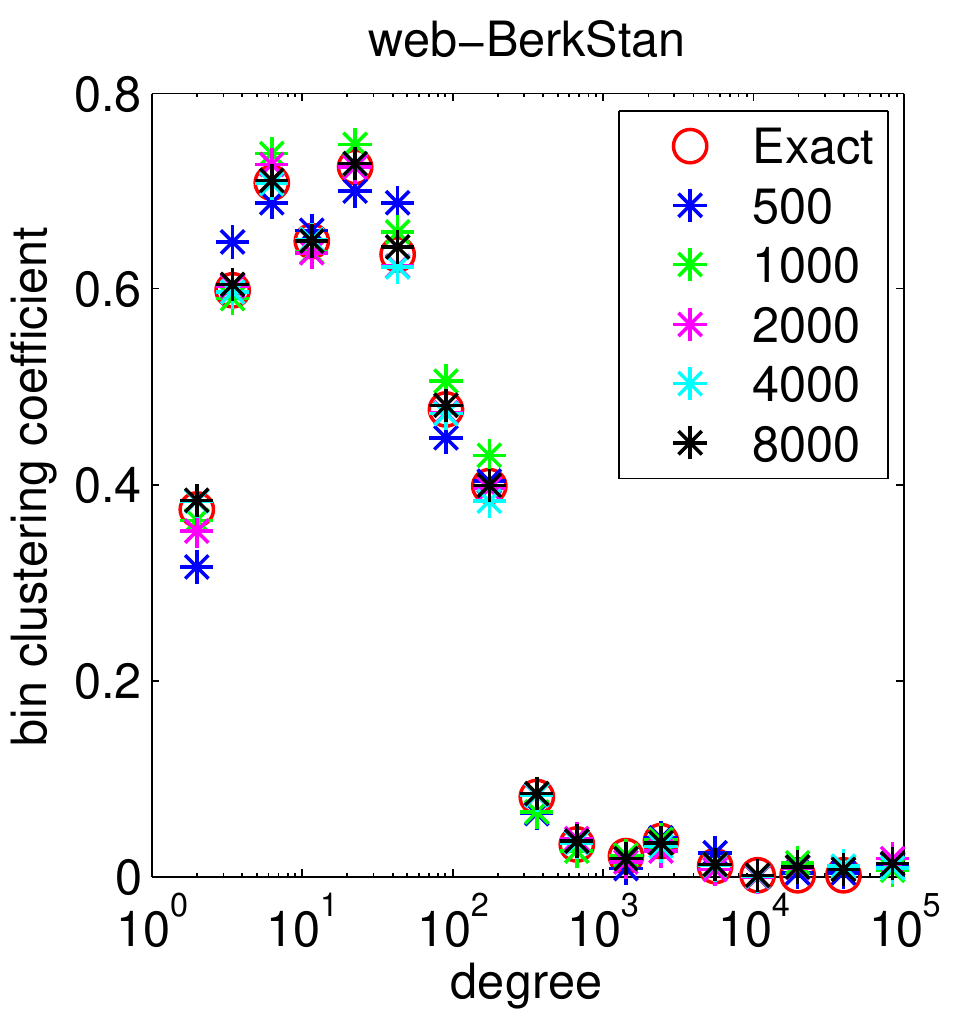}
  }
    \subfloat{\label{fig:ccd:web-Google}
    \includegraphics[width=2in,trim=0 0 0 0]{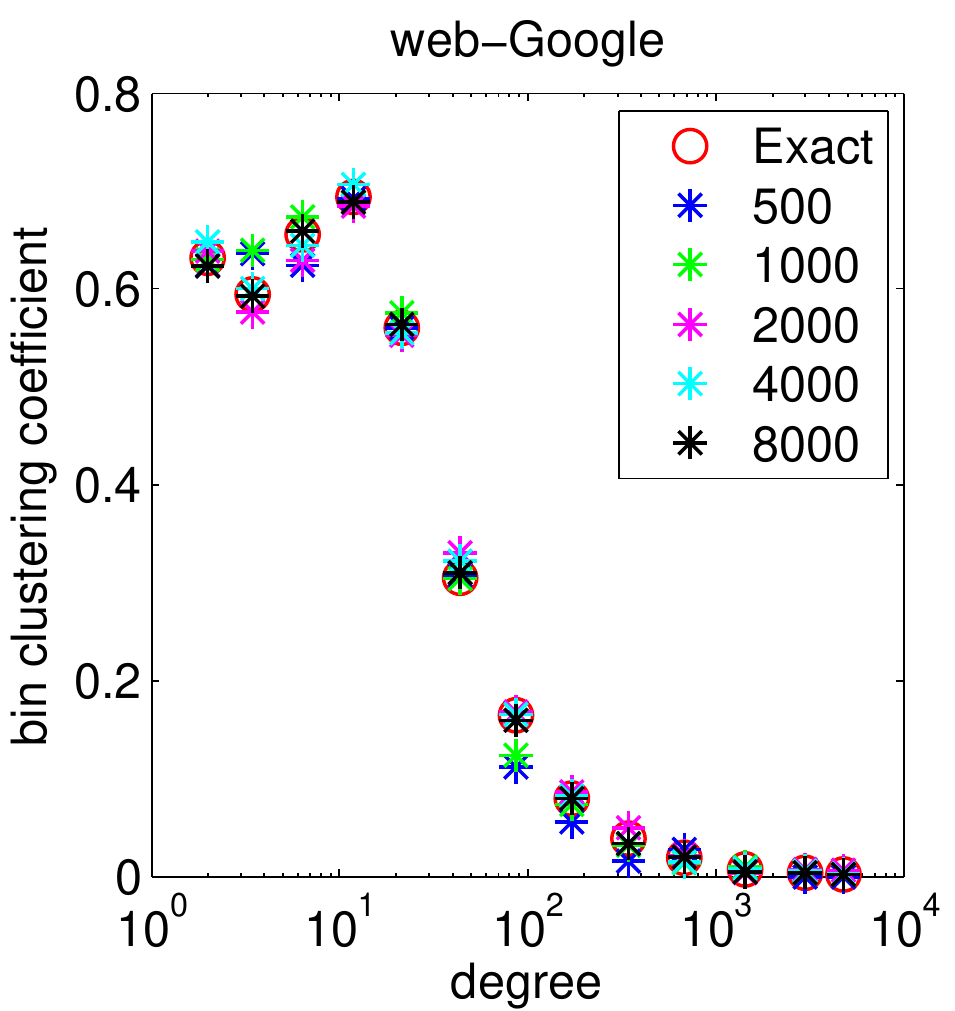}
  }
  \\
  \subfloat{\label{fig:ccd:web-Stanford}
    \includegraphics[width=2in,trim=0 0 0 0]{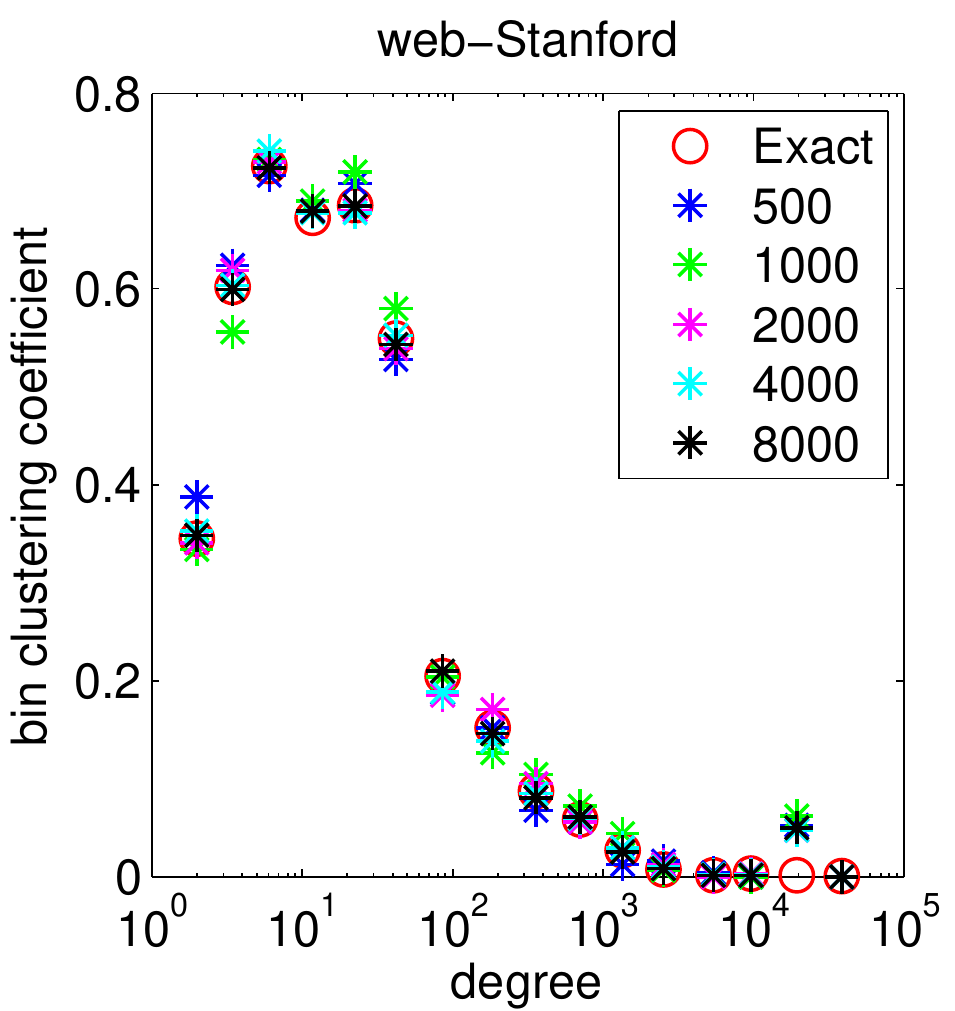}
  }
  \subfloat{\label{fig:ccd:wiki-Talk}
    \includegraphics[width=2in,trim=0 0 0 0]{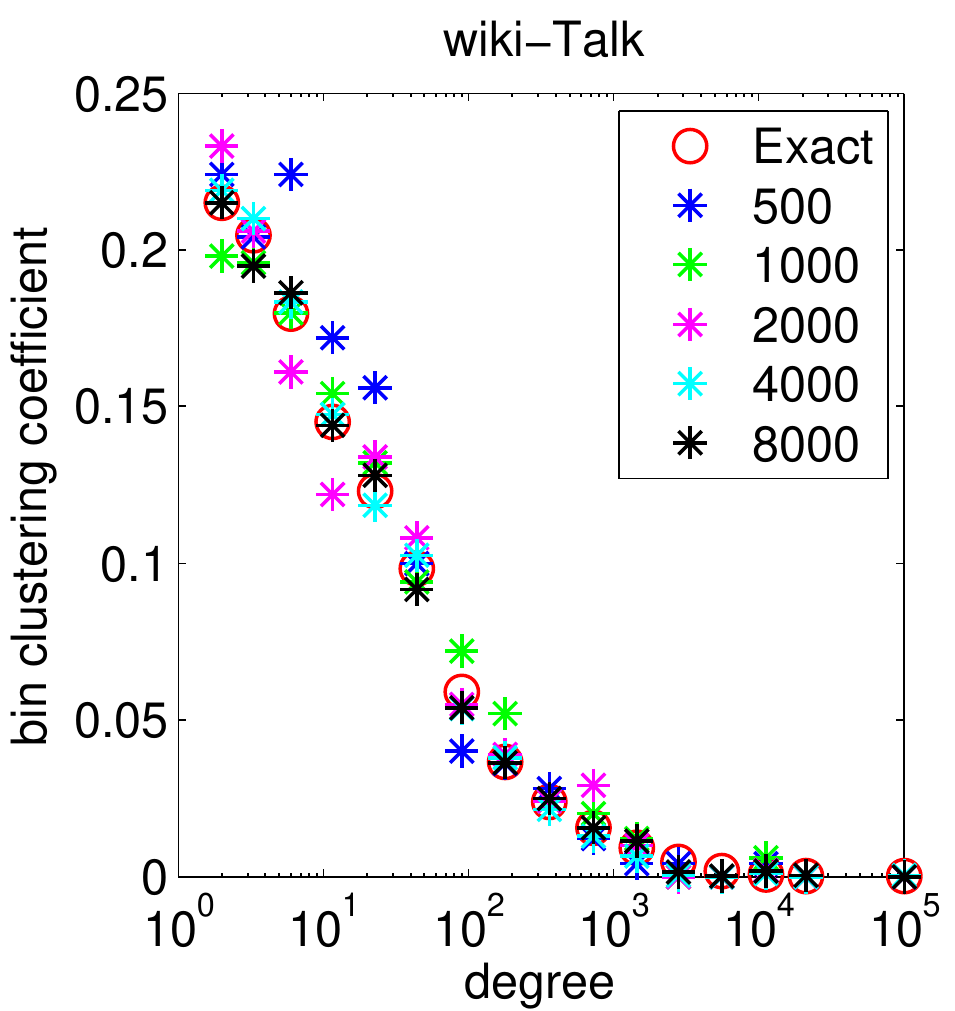}
  }    

\caption{Computing degree-wise clustering coefficients using wedge sampling}
\label{fig:ccd-app}
\end{figure*}

\clearpage

\begin{table*}[p]
 \caption{Comparison of wedge sampling  and Doulion for computing  the global clustering coefficient }
\label{tab:Dgcc}  
  \centering\small
  \begin{tabular}{|>{\tt}r@{\,}||r|rrr|rrr|rrr|rrr|r|}
    \cline{3-15}
    \multicolumn{2}{c|}{ } 
    & \multicolumn{3}{c|}{Wedge Sampling}
    & \multicolumn{3}{c|}{Doulion 1/25}
    & \multicolumn{3}{c|}{Doulion 1/50}
    & \multicolumn{3}{c|}{Doulion 1/100}
    & \multicolumn{1}{c|}{Time } \\
    \hline
    \hline
\cta{Graph} &\ct {$\gcc$} & \ct{min} &\ct{max}&  \ct{sd}&  \ct{min} &\ct{max}&  \ct{sd}&  \ct{min} &\ct{max}&  \ct{sd}&  \ct{min} &\ct{max}&  \ct{sd}  & \ct{\small D50/WS}  \\  
    \hline
          amaz0312  &    .160&   .155   & .166   & .002&  . .137  &  .188  &  .011  &  .093  &  .234  &  .031& .000   & .392  &  .080 & 7.06\\
             amaz0505 &     .162&  .159  &  .167    &.002     &.133  &  .193    &.011   & .098 &   .241  &  .028 &  .000   & .370   & .081 & 7.37 \\
             amaz0601 &     .166&  .162 &   .172  &  .002    &  .140   & .193 &   .010 &   .088   & .260   & .028  & .000    &.457 &   .078 & 7.30\\
             skitter         &      .005 & .005  &  .007 &   .000   &   .005  &  .006   & .000  &  .004  &  .007  &  .001&  .002    &.008  &  .001 & 17.80	\\
                     cit-Pat &     .067&  .064    &.071    &.001   &  .060    &.077  &  .003   & .042    &.099   & .010 & .027   & .116  &  .022 & 15.76\\
                 road-CA &     .060&  .058    &.065  &  .001   &       .000  &  .133   & .023    &     .000 &   .250  &  .062  & .000 &   1.001  &  .193 & 5.84 \\
      	      w-BerSta &     .007&  .006    &.008   & .001   &   .006   & .008   & .000    &.006   &.008   & .000 & .004  &  .011  &  .001  &16.28\\
              w-Google &     .055&   .052  &  .059   & .001    &  .051  &  .060   & .002   & .044 &   .071  &  .005 &  .021   & .107   & .016 &1.20\\
                  w-Stan &     .009&    .007   & .010 &   .001   & .007 &   .010 &   .001  &  .006  &  .012   & .001 &  .003  &  .018  &  .003 &7.18 \\
                   wiki-T &     .002	&     .002   & .003   & .000   &  .002  &  .002&    .000   & .002   & .003  &  .000  & .001  &  .004  &  .001 &8.52 \\
               youtube &     .006&     .005    &.007   & .000     & .005   & .007 &   .001  &  .004  &  .010  &  .001&  .000    & .018  &  .004 & 7.55\\
               flickr       &     .112 &    .108   & .117 &   .002    &  .110  &  .115   & .001 &   .107 &   .118 &   .002 & .098 &   .125   & .005 &11.38 \\
              livejour  &    .124 &      .120   & .128   & .002   & .121 &   .128  &  .001   & .116  &  .130  &  .003 & .105  &  .143   & .007 &3.49\\
            \hline

 \end{tabular}
\end{table*}

\begin{table*}[p]
 \caption{Comparison of wedge sampling  and Doulion for computing  the local clustering coefficient }
\label{tab:Dlcc}  
  \centering\small
  \begin{tabular}{|>{\tt}r@{\,}||r|rrr|rrr|rrr|rrr|}
    \hline 
    \multicolumn{1}{|c|}{ } 
&     \multicolumn{1}{c|}{ } 
    & \multicolumn{3}{c|}{Wedge Sampling}
    & \multicolumn{3}{c|}{Doulion 1/10}
    & \multicolumn{3}{c|}{Doulion 1/25}
    & \multicolumn{3}{c|}{Doulion 1/50}\\
   \cline{3-14}
            \cta{Graph} &\ct {$\lcc$} & \ct{min} &\ct{max}&  \ct{stdev}&  \ct{min} &\ct{max}&  \ct{stdev}&  \ct{min} &\ct{max}&  \ct{stdev}&  \ct{min} &\ct{max}&  \ct{stdev}   \\  
    \hline       
             amaz0312  &	.421   &  .415   &  .426   &  .003   &  .395    & .463  &   .014  &   .318   &  .558  &   .047   &  .195   & .869     &.140 \\
             amaz0505 &    .427    & .421    & .432    & .003     &.395   &  .463   &  .014    & .338    & .568  &   .056    & .195	 & .869    & .140 \\
             amaz0601 &      .430    & .423     &.437    & .003   &  .403     &.466   &  .013  &   .329   &  .633  &   .058 &    .239  &.819    & .120 \\
             skitter &     .296   & .288   &  .303   &  .003   &  .272   &  .322   &  .011    & .206  &   .384    & .039    & .122  &.667     &.105 \\
                        cit-Pat &  .092   &  .088   &  .096     &.002   &  .085    & .099    & .003   &  .066    & .126   &  .011   &  .028  & .199    & .035 \\
            	             road-CA &          .055  &   .052   &  .058   &  .001   &  .038    & .071    & .006     &     .000    & .118  &   .022    &      .000  	  &  .279   &  .062 \\
	                        w-BerSta &   .634  &   .627   &  .641   &  .003     &.586   &  .700    & .024   &  .532    & .808    & .057     &.339   & 1.273  &   .165 \\
             w-Google &     .624   &  .615  &   .630 &    .003  &   .580   &  .688    & .021    & .471    & .772   &  .063     &.335   &1.276    & .183 \\
           w-Stan &      .629    & .622  &   .636    & .003    & .532  &   .737   &  .038     &.441    & .982   &  .109  &   .218  &1.218     &.230 \\
              wiki-T &  .201   & .195   &  .206  &   .002   &  .171    & .240    & .015   &  .055   &  .379   &  .067   &  .006   &.609  &   .160 \\
              youtube &   .128    & .167     &.179     &.002 &    .128     &.218  &   .015   &  .055    & .294  &   .057    & .007  &.767  &   .182 \\
              flickr &  .375   &  .368     &.381    & .003    & .316   &  .411     &.016     &.212    & .553    & .066  &   .086  &  .774    & .154 \\
              livejour &  .345    & .337    & .355   &  .003   &  .330   &  .359     &.006    & .296  &   .401  &   .023  &   .214   & .500    & .060 \\    
                          
    \hline
 \end{tabular}
\end{table*}

\end{document}